\documentclass[final,5p,times,twocolumn]{elsarticle}

\usepackage{amssymb}
\usepackage{graphicx}
\usepackage[cmex10]{amsmath}
\usepackage{array}
\usepackage{mdwmath}
\usepackage{mdwtab}
\usepackage{fixltx2e}
\usepackage{graphics} % for pdf, bitmapped graphics files
\usepackage{epsfig} % for postscript graphics files
\usepackage{times} % assumes new font selection scheme installed

\usepackage{algorithm,algorithmic}
\usepackage{url}
\usepackage{lscape}
\usepackage{color}
\usepackage{epsfig}

\usepackage{algorithm,tabularx}
\usepackage[utf8x]{inputenc}
\usepackage{hyperref}
\usepackage{graphicx}
\usepackage[cmex10]{amsmath}
\usepackage{array}
\usepackage{graphics}

\usepackage{amsthm}

\usepackage{url}
\usepackage{lscape}
\usepackage{color}
\usepackage{epsfig}
\newtheorem{theorem}{Theorem}
\newtheorem{remark}{Remark}
\newtheorem{lemma}{Lemma}

\newtheorem{problem}{Problem}

\newtheorem{assumption}{Assumption}

\usepackage{xparse}
\usepackage{etoolbox}

\newtheorem{lemma-non}{Lemma}

\usepackage{balance}
\usepackage{dblfloatfix}
\usepackage{nicefrac}

\journal{Control Engineering Practice}

\begin{document}

\begin{frontmatter}

\title{Discrete-Time Implementation of Explicit Reference Governor\tnoteref{Funding}\tnoteref{Funding2}}

\tnotetext[Funding]{This research has been supported by WSU Voiland College of Engineering and Architecture through a start-up package to M. Hosseinzadeh.}
\tnotetext[Funding2]{The authors declare that there is no conflict of interest regarding the publication of this article.}

\author{Mu'taz A. Momani}
\ead{mutaz.momani@wsu.edu}

\author{Mehdi Hosseinzadeh\corref{cor}}
\ead{mehdi.hosseinzadeh@wsu.edu}

\address{School of Mechanical and Materials Engineering, Washington State University, Pullman, WA 99164, USA}

\cortext[cor]{Corresponding author.}

\begin{abstract}
Explicit reference governor (ERG) is an \textit{add-on} unit that provides constraint handling capability to pre-stabilized systems. The main idea behind ERG is to manipulate the derivative of the applied reference in continuous time such that the satisfaction of state and input constraints is guaranteed at all times. However, ERG should be practically implemented in discrete-time. This paper studies the discrete-time implementation of ERG, and provides conditions under which the feasibility and convergence properties of the ERG framework are maintained when the updates of the applied reference are performed in discrete time. The proposed approach is validated via extensive simulation and experimental studies. 
\end{abstract}


\end{frontmatter}

%%%%%%%%%%%%%%%%%%%%%%%%%%%%%%%%%%%%%%%%%%%

\section{Introduction}
The satisfaction of constraints (e.g., safety requirements and actuator saturations) is a crucial requirement for the control of many real-world systems. This points out the critical need for constrained control schemes, where the constraint satisfaction is enforced at all times, while addressing performance and control objectives.

In the last few decades, the literature on constrained control has been dominated by optimization-based schemes, more specifically, Model Predictive Control (MPC) \cite{Camacho2013,Rawlings2017}, Control Barrier Function (CBF) \cite{Ames2019,Cortez2021}, and Reference/Command Governor (RG/CG) \cite{Garone2016,Hosseinzadeh2022_ROTEC}. These methods have the remarkable feature of being able to optimize various measures of performance objectives while guaranteeing the satisfaction of constraints at all times. However, the use of online optimization precludes the application of those methods to safety-critical systems where the onboard computing power is limited and there are stringent software reliability/certification requirements.

Recently, a novel scheme called explicit RG (ERG) has been proposed for the control of constrained systems \cite{nicotra2018explicit,HosseinzadehECC,Hosseinzadeh2020,hosseinzadeh2019explicit,Hosseinzadeh2020_CEP,HosseinzadehLetter,Garone2018,Nicotra2015,Garone2016_ERG}. The ERG is a framework inspired by the RG philosophy, whose main novelty is that it does not require any online optimization. The main idea behind the ERG is to manipulate the derivative of the applied reference in continuous time using a suitable nonlinear control law, ensuring that constraints are satisfied at all time.

Given its generality and implementation simplicity, the ERG has the potential to become a common method for the constrained control of those systems in which the use of optimization-based methods is not realistic. At the present stage, one of the main obstacles to a widespread diffusion of the ERG is that it has been introduced as a continuous-time scheme and its theoretical properties have been derived under the assumption that ERG updates the applied reference in continuous time. However, ERG should be practically implemented in discrete-time, and its not clear if  its properties are maintained when the applied reference is updated in discrete time. Note that prior work involving computer implementation of ERG (see, e.g., \cite{Merckaert2022,Nicotra2019_Delay,HosseinzadehMED}) have not provided implementation details, and to the best of our knowledge, discrete-time implementation of ERG has not been carefully investigated in the literature.

This paper studies the theoretical guarantees of ERG when the updates of the applied reference are computed in discrete time, and provides conditions under which the properties of the ERG framework are maintained. Extensive simulation and experimental studies have been provided to support the analyses.

The rest of this paper is organized as follows. Section \ref{sec:PS} states the problem. Section \ref{sec:DTI} proposes the discrete-time ERG, and analyzes its properties. Section \ref{sec:MainResults} outlines the discrete-time ERG as a main result of the paper, highlighting its features and providing a pseudocode. Extensive numerical simulations have been reported in Section \ref{sec:NA}. Section \ref{sec:ExperimentalResults} presents the experimental validation of the proposed discrete-time ERG using a Parrot Bebop 2 drone. Finally, Section \ref{sec:Conclusion} concludes the paper.

\paragraph*{Notation} We denote the set of real numbers by $\mathbb{R}$, the set of positive real numbers by $\mathbb{R}_{>0}$, and the set of non-negative real numbers by $\mathbb{R}_{\geq0}$. We use $\mathbb{Z}_{\geq0}$ to denote the set of non-negative integer numbers. We denote the transpose of matrix $A$ by $A^\top$. For a function $Y(x)$, its gradient with respect to $x$ is denoted as $\nabla_xY(x)$; also, $Y(x)|_{x=x^\dag}$ indicates that the function $Y(x)$ is evaluated at $x=x^\dag$. We use $t$ to denote continuous time and $k$ to denote sampling instants. Given $x\in\mathbb{R}^n$,  $\left\Vert x\right\Vert^2=x^\top x$. Finally, $\mathbf{0}$ denotes the zero matrix with appropriate dimension.

%%%%%%%%%%%%%%%%%%%%%%%%%%%%%%%%%%%%%%%%%%%

\section{Problem Statement}\label{sec:PS}
Consider the following pre-stabilized system
\begin{align}\label{eq:system}
\dot{x}(t)=f\big(x(t),v(t)\big),
\end{align}
where $x(t)=[x_1(t)~\cdots~x_n(t)]^\top\in\mathbb{R}^n$ is the state vector and $v(t)=[v_1(t)~\cdots~v_m(t)]^\top\in\mathbb{R}^m$ is the applied reference signal. By pre-stabilized, we mean that for any constant applied reference $v$, system \eqref{eq:system} asymptotically converges to the equilibrium point $\bar{x}_v$ satisfying $f(\bar{x}_v,v)=0$. The pre-stabilization of unconstrained systems is the subject of an extensive literature (see, e.g., \cite{Ogata,Khalil,Isidori1995}), and can be addressed using a variety of available techniques.

\begin{assumption}\label{Assumption:Lipchitz}
There exists $\mu\in\mathbb{R}_{>0}$ such that $\left\Vert\nabla_v\bar{x}_v\right\Vert\leq\mu$. Note that this assumption is reasonable as it corresponds to a bound on the sensitivity of $\bar{x}_v$ to $v$.
\end{assumption}

Since system \eqref{eq:system} is pre-stabilized, the converse theorem \cite[Section 3.6]{Khalil} implies that for any constant applied reference $v$, there exists a Lyapunov function $V\big(x,v\big)$ that proves the asymptotic stability of the equilibrium point $\bar{x}_v$. Suppose the Lyapunov function $V\big(x,v\big)$ has the following property, for some $m_1,m_2\in\mathbb{R}_{>0}$:
\begin{align}\label{eq:Lyapunov}
m_1\left\Vert x-\bar{x}_{v}\right\Vert^2\leq V\left(x,v\right)\leq m_2\left\Vert x-\bar{x}_{v}\right\Vert^2.
\end{align}
.

Let system \eqref{eq:system} be subject to the following constraints:
\begin{align}\label{eq:constraints}
c_i(x,v)\geq0,~i=1,\cdots,n_c
\end{align}
where $c_i:\mathbb{R}^n\times\mathbb{R}^m\rightarrow\mathbb{R}$ and $n_c\in\mathbb{R}_{\geq0}$ indicates the number of constraints. For any constant applied reference $v$, the set $\{x|c_i\left(x,v\right)\geq0,~i\in\{1,\cdots,n_c\}\}$ is convex\footnote{When the set $\{x|c_i\left(x,v\right)\geq0,~i\in\{1,\cdots,n_c\}\}$ is non-convex, one can, first, use the techniques described in \cite{Koditschek1990,Nunez2008,Ji2015} to distort the set into a convex one, and then use the techniques described in \cite{nicotra2018explicit} to build an ERG. Future work will consider more general constraints.}.

Given the desired reference $r\in\mathbb{R}^m$, we use the ERG framework to determine the applied reference $v(t)$ that steers the states of the system to $\bar{x}_r$ (or to an admissible approximation of $\bar{x}_r$). The ERG manipulates the applied reference $v(t)$ according to the following continuous-time nonlinear law:
\begin{align}\label{eq:ERGClassic}
\dot{v}(t)=\kappa\cdot g(x(t),v(t),r),
\end{align}
where $\kappa\in\mathbb{R}_{>0}$ is an arbitrary scalar, and $g(x,v,r):=\Delta(x,v)\cdot\rho(v,r)$ with $\Delta(x,v)$ being the Dynamic Safety Margin (DSM) and $\rho(v,r)$ being the Attraction Field (AF); see Figure \ref{fig:erg} for a schematic structure of ERG.

As discussed in \cite{hosseinzadeh2019explicit,nicotra2018explicit}, one possible way to build a DSM for system \eqref{eq:system} is to use Lyapunov theory\footnote{Note that even though this paper considers a Lyapunov-based DSM, the presented analyses are true when ERG is implemented with a trajectory-based DSM; see \cite{nicotra2018explicit,Hosseinzadeh2020} for details on trajectory-based DSMs.}. Given a constant $v\in\mathbb{R}^m$, let $\Gamma\big(v\big)$ be a threshold value computed by solving the following optimization problem:
\begin{align}\label{eq:OptimizationLyapunovBasedERG}
\Gamma\left(v\right)=\left\{
\begin{array}{cc}
     &  \min\limits_{z\in\mathbb{R}^n} V(z,v)\\
    \text{s.t.} & c_i(z,v)\leq0,~\forall i
\end{array}
\right..
\end{align}

It is obvious that $V\big(x(t),v\big)\leq\Gamma\big(v\big)$ implies that $c_i\left(x(t),v\right)\geq0,~\forall i$ at all times. Then, a DSM can be defined as
\begin{align}\label{eq:DSM}
\Delta\big(x(t),v\big)=\Gamma\big(v\big)-V\big(x(t),v\big). 
\end{align}

\begin{remark}
Given a constant $v$, any $\hat{\Gamma}\left(v\right)\leq\Gamma\left(v\right)$ is a sub-optimal threshold value \cite{nicotra2018explicit}. Thus, according to \eqref{eq:Lyapunov}, the following optimization problem gives a sub-optimal threshold value:
\begin{align}\label{eq:GammaOpt1}
\hat{\Gamma}\left(v\right)=\left\{
\begin{array}{cc}
     &  \min\limits_{z\in\mathbb{R}^n} m_1\left\Vert z-\bar{x}_{v}\right\Vert^2\\
    \text{s.t.} & c_i(z,v)\leq0,~\forall i
\end{array}
\right..
\end{align}
\end{remark}

Regarding the AF, it can be designed \cite{HosseinzadehMED} by decoupling it into an attraction term and a repulsion term, as follows: 
\begin{align}\label{eq:AF}
\rho(v,r)=\rho_a(v,r)+\rho_r(v),
\end{align}
where the attraction term $\rho_a(v,r)$ is a vector field which points towards
the desired reference, and the repulsion term $\rho_r(v)$ is a vector field which points away from the constraints. These two vector fields can be defined as
\begin{align}\label{eq:AFAttraction}
\rho_a(v,r)=\frac{r-v}{\max\{\left\Vert r-v\right\Vert,\eta_1\}},
\end{align}
and 
\begin{align}\label{eq:AFRepulsive}
\rho_r(v)=&\sum\limits_{i=1}^{n_c}\max\left\{\frac{\xi-c_i(\bar{x}_v,v)}{\xi-\delta},0\right\}\frac{\nabla_vc_i(\bar{x}_v,v)}{\left\Vert\nabla_vc_i(\bar{x}_v,v)\right\Vert},
\end{align}
where $\eta_1\in\mathbb{R}_{>0}$ is a smoothing factor, and $\xi,\delta\in\mathbb{R}_{>0}$ satisfying $\xi>\delta$ are design parameters.

\begin{figure}
    \centering
    \includegraphics[width=\linewidth]{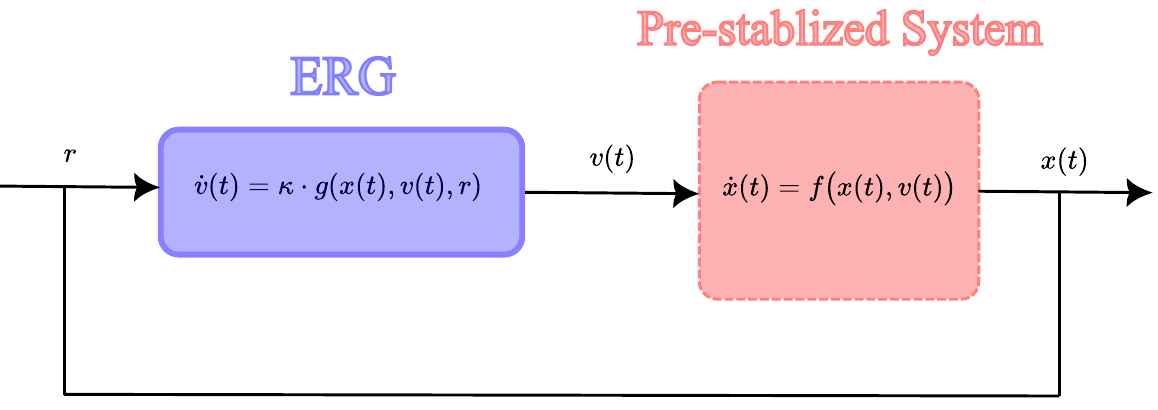}
    \caption{Schematic structure of the ERG framework.}
    \label{fig:erg}
\end{figure}

This paper considers the following problem. 

\begin{problem}\label{Problem}
Let ERG given in \eqref{eq:ERGClassic} be implemented in discrete time by using the Euler method, as follows:
\begin{align}\label{eq:ERGDiscrete}
v(kdt)=v((k-1)dt)+dt\cdot\kappa\cdot g(x(kdt),v((k-1)dt),r),
\end{align}
where $dt$ is the sampling period (i.e., $[(k-1)dt,kdt)$ is equal to $dt$ seconds), and $k\in\mathbb{Z}_{\geq0}$. Equation \eqref{eq:ERGDiscrete} indicates that the applied reference $v(t)$ is updated every $dt$ seconds based on the current value of the state and the previous value of the applied reference. Given the desired reference $r\in\mathbb{R}^m$, provide conditions under which a zero-order hold implementation of the applied reference (i.e., $v(t)=v(kdt)$ for $t\in[kdt,(k+1)dt)$) satisfies the following properties: i) constraints \eqref{eq:constraints} are satisfied at all time (i.e., $c_i(x(t),v(t))\geq0,~\forall i$ for $t\geq0$); ii) if $r$ is \textit{strictly} steady-state admissible (i.e., $c_i(\bar{x}_{r},r)\geq\xi,~\forall i$), then $v(kdt)\rightarrow r$ as $k\rightarrow\infty$; and iii) if $r$ is not strictly steady-state admissible, then $v(kdt)\rightarrow r^\ast$ as $k\rightarrow\infty$, where $r^\ast$ is the \textit{best} admissible approximation of $r$ satisfying $\delta \leq c_i(\bar{x}_{r^\ast},r^\ast) < \xi$ for some $i$. 
\end{problem}

%%%%%%%%%%%%%%%%%%%%%%%%%%%%%%%%%%%%%%%%%%%
\section{Theoretical Analysis}\label{sec:DTI}
When ERG \eqref{eq:ERGClassic} is implemented in continuous time with DSM \eqref{eq:DSM} and AF \eqref{eq:AF}, the repulsive vector field given in \eqref{eq:AFRepulsive} ensures that $c_i\left(\bar{x}_{v(t)},v(t)\right)\geq\delta,~i=1,\cdots,n_c$ for all $t$. Thus, since  $\delta\in\mathbb{R}_{>0}$, it can be concluded that the set $\mathcal{D}$ given below\textemdash which represents the set of steady-state admissible equilibria\textemdash is invariant:
\begin{align}\label{eq:setD}
\mathcal{D}=\left\{v|c_i(\bar{x}_v,v)\geq\delta,~i=1,\cdots,n_c\right\}.
\end{align}

However, when the updates of applied reference $v$ are computed in discrete time as in \eqref{eq:ERGDiscrete}, discretization errors can violate the invariance of the set $\mathcal{D}$; see Figure \ref{fig:DiscretizationError} for a geometric illustration. Note that one can address this issue by using a large $\delta$. This approach is not desired as using a large $\delta$ implies that safe equilibria could be treated unsafe, which can hamper desirable properties of ERG.

% at the cost of performance degradation; the larger the $\delta$ is, the poorer the performance of the ERG could become, as safe equilibria would be treated unsafe. 

\begin{figure}
    \centering
    \includegraphics[width=\linewidth]{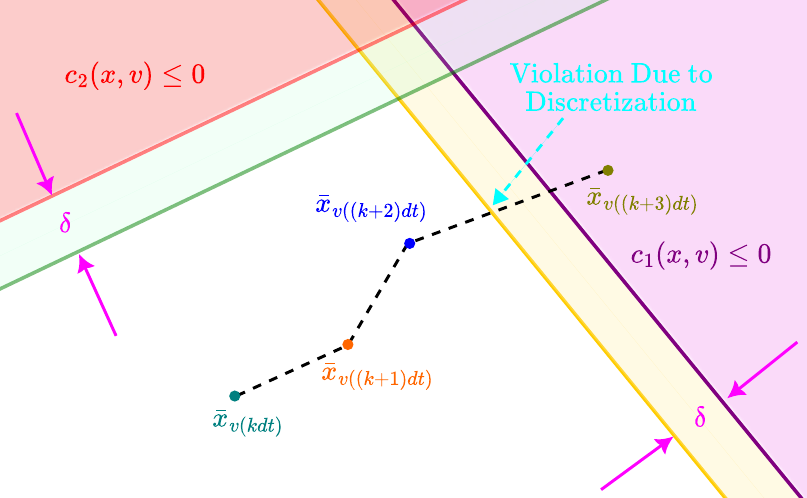}
    \caption{Violation of the invariance of the set $\mathcal{D}$ due to discretization errors.}
    \label{fig:DiscretizationError}
\end{figure}

From \eqref{eq:ERGDiscrete}, another approach to avoid violations due to discretization errors is to limit the change of the applied reference $v(t)$, but at the costs of performance degradation. Limiting $\big\Vert v(kdt)-v((k-1)dt)\big\Vert$ can be done by either using a saturation function or selecting $dt$ and $\kappa$ such that $dt\cdot\kappa$ is sufficiently small at all times. Using a saturation function is limited to simulation-based studies, e.g., \cite{hosseinzadeh2019explicit,HosseinzadehLetter}, as: i) it is unclear how one can determine a safe saturation level with minimum performance degradation; ii) a safe saturation level for one operating scenario might be unsafe for other scenarios; and iii) when constraints are far from being violated, using a saturation function can unnecessarily degrade the performance.

Regarding the second approach (i.e., selecting $dt$ and $\kappa$ such that $dt\cdot\kappa$ is sufficiently small), prior work (e.g., \cite{HosseinzadehMED,Hosseinzadeh2022_ROTEC,Hosseinzadeh2023RobustTermination}) mainly focuses on using a small sampling period $dt$. Note that we cannot arbitrarily reduce the sampling period $dt$, as: i) small $dt$ requires frequent computations which can be challenging for real-time scheduling; see, e.g., \cite{Roy2021,Wilhelm2008}; and ii) there is a minimum time required for computing the applied reference via \eqref{eq:ERGDiscrete}, which imposes a lower limit on the sampling period $dt$. Using a small $\kappa$ is not a reliable and efficient approach, as: i) it is unclear how the value of $\kappa$ should be determined to avoid violations due to discretization errors (safe $\kappa$ for one operating scenario might be unsafe for a different scenario); and ii) when large changes in the applied reference are allowed, using a small constant $\kappa$ can unnecessarily prevent that and consequently degrade the performance.

Instead of using a constant $\kappa$, this paper proposes using a dynamic $\kappa$ whose value at any sampling instant $kdt$ is selected such that constraint satisfaction and convergence are guaranteed at all times, with minimum performance degradation. For this purpose, first, the following lemma provides a condition on the design parameter $\kappa$, under which the invariance of the set $\mathcal{D}$ given in \eqref{eq:setD} is guaranteed at all times, when discrete-time ERG given in \eqref{eq:ERGDiscrete} is utilized to compute the updates of the applied reference.

% What is critically overlooked in the prior work is the role of the design parameter $\kappa$. In what follows, we investigate the impact of $\kappa$ on feasibility and convergence properties of the discrete-time ERG given in \eqref{eq:ERGDiscrete}. 

% A potential alternative is to reduce the increments by decreasing the design parameter $\kappa$ rather than $\Delta{t}$. Given that $\kappa$ is a design parameter, we may decrease it as much as we wish. However, even with this adjustment, there remains a trade-off as the increments may still be small.

\begin{lemma}\label{Lemma:Condition}
Let $v((k-1)dt)\in\mathcal{D}$. At sampling instant $kdt$, let the design parameter $\kappa(kdt)$ be selected such that the following inequality is satisfied:
\begin{align}\label{eq:ConditionKappa}
\kappa(kdt)\leq\frac{\bar{\vartheta}}{\mu \cdot dt\cdot \max\left\{\left\Vert g(x(kdt),v((k-1)dt),r)\right\Vert,\eta_2\right\}},
\end{align}
where $\eta_2\in\mathbb{R}_{>0}$ is a smoothing factor, and $\bar{\vartheta}=\min_{i\in\{1,\cdots,n_c\}}\{\vartheta_i\}$ with $\vartheta_i\in\mathbb{R}_{\geq0}$ being the Euclidean distance between the equilibrium point $\bar{x}_{v((k-1)dt)}$ and the $i$th tightened constraint, i.e., 
\begin{align}\label{eq:Distance}
\vartheta_i=\left\{
\begin{array}{ll}
     & \min\,\left\Vert \bar{x}_{v((k-1)dt)}-\theta\right\Vert \\
    \text{s.t.} & c_i(\theta,v((k-1)dt))=\delta
\end{array}
\right..
\end{align}

Then, $v(kdt)\in\mathcal{D}$.
\end{lemma}

\begin{proof}
On the one hand, Assumption \ref{Assumption:Lipchitz} establishes a relationship that connects the changes in the equilibrium points to variations in the applied reference, as follows:
\begin{align}
\left\Vert\bar{x}_{v(kdt)} - \bar{x}_{v((k-1)dt)}\right\Vert &\leq\mu\left\Vert v(kdt)-v((k-1)dt)\right\Vert,
\end{align}
which according to \eqref{eq:ERGDiscrete} implies that 
\begin{align}
&\left\Vert\bar{x}_{v(kdt)} - \bar{x}_{v((k-1)dt)}\right\Vert\leq\mu\cdot dt\cdot\kappa\cdot\left\Vert g(x(kdt),v((k-1)dt),r)\right\Vert\nonumber\\
&\leq\mu\cdot dt\cdot\kappa\cdot\max\left\{\left\Vert g(x(kdt),v((k-1)dt),r)\right\Vert,\eta_2\right\}.
\label{eq:relationship2}
\end{align}

% \begin{figure*}[b]
% \hrule
% \hrule\setcounter{equation}{17}
% \begin{align}\label{eq:ConditionKappa2}
% \kappa(kdt)\leq\frac{\max\{\frac{\sqrt{m_1}}{\sqrt{m_1}+\sqrt{m_2}}\bar{\vartheta}-\frac{\sqrt{m_2}}{\sqrt{m_1}+\sqrt{m_2}}\left\Vert x(kdt)-\bar{x}_{v((k-1)dt)}\right\Vert,0\}}{\mu \cdot dt\cdot \max\left\{\left\Vert g(x(kdt),v((k-1)dt),r)\right\Vert,\eta_2\right\}}.
% \end{align}   
% \end{figure*}

On the other hand, given $v((k-1)dt)\in\mathcal{D}$, imposing $\big\Vert\bar{x}_{v(kdt)} - \bar{x}_{v((k-1)dt)}\big\Vert\leq\bar{\vartheta}$ guarantees that $v(kdt)\in\mathcal{D}$, where $\bar{\vartheta}$ is the smallest distance between the equilibrium point $\bar{x}_{v((k-1)dt)}$ and the constraints. See Figure \ref{fig:Condition} for a geometric illustration.

Thus, no further effort is needed to show that if $\kappa$ satisfies condition \eqref{eq:ConditionKappa} at sampling instant $kdt$, then $v(kdt)\in\mathcal{D}$.
\end{proof}

\begin{figure}
    \centering
    \includegraphics[width=\linewidth]{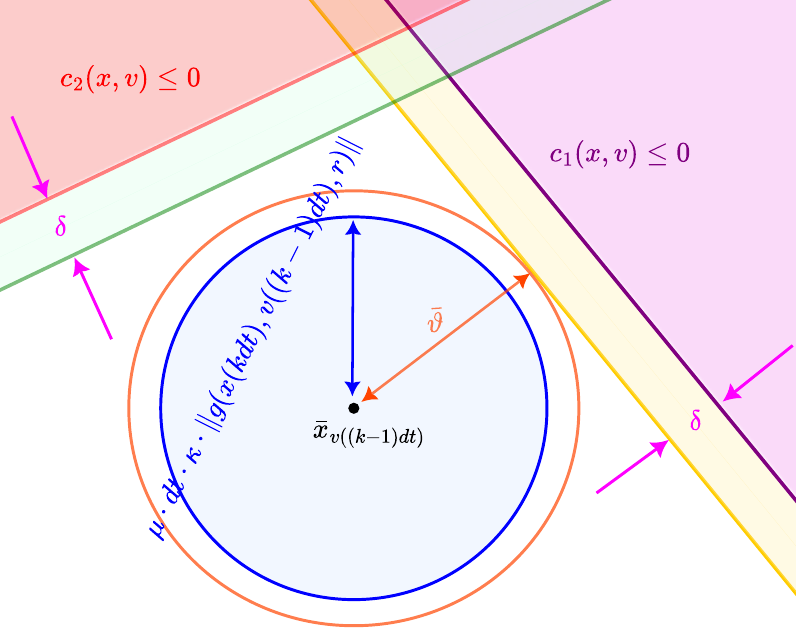}
    \caption{Geometric illustration of the impact of the design parameter $\kappa$ on the invariance of the set $\mathcal{D}$.}
    \label{fig:Condition}
\end{figure}

\begin{remark}\label{remark:SteadyStateDelta}
Starting from the initial condition $v(0)\in\mathcal{D}$, Lemma \ref{Lemma:Condition} indicates that if the updates of the applied reference are obtained via \eqref{eq:ERGDiscrete}, the applied reference remains steady-state admissible at all times; that is, the set $\mathcal{D}$ is invariant. Thus, it can be concluded \cite{nicotra2018explicit,hosseinzadeh2019explicit} that there exits $\epsilon\in\mathbb{R}_{>0}$ such that $\Delta\left(\bar{x}_{v(kdt)},v(kdt)\right)\geq\epsilon$; this property implies that $\Delta\left(x(t),v(t)\right)$ cannot remain equal to zero. 
\end{remark}

% The following theorems show that discrete-time ERG given in \eqref{eq:ERGDiscrete} ensures convergence to the desired reference $r$, as well as constraint satisfaction at all times. 

The following theorem shows that if the updates of the applied reference are obtained as in \eqref{eq:ERGDiscrete}, $v(kdt)$ asymptotically converges to $r$ if it is strictly steady-state admissible, or else to the \textit{best} admissible approximation of $r$.

\begin{theorem}\label{theorem:Convergence}
Consider the pre-stabilized system \eqref{eq:system} which is subject to constraints \eqref{eq:constraints}. Suppose that the discrete-time ERG \eqref{eq:ERGDiscrete} is utilized to obtain the applied reference at each sampling instant $kdt$. Let the initial condition $v(0)$ be such that $c_i(\bar{x}_{v(0)},v(0))\geq\delta,~\forall i$. Then, $v(kdt)\rightarrow r$ as $k\rightarrow\infty$ if $r$ is strictly steady-state admissible, or else $v(kdt)\rightarrow r^\ast$ as $k\rightarrow\infty$, where $r^\ast$ is the best admissible approximation of $r$ satisfying $\delta \leq c_i(\bar{x}_{r^\ast},r^\ast) < \xi$ for some $i$.
\end{theorem}

\begin{proof}
Consider the following Lyapunov function:
\begin{align}
W\left(v(kdt)\right)=\sum_{\omega=v(kdt)}^{r}\left\Vert \rho\left(\omega,r\right)\right\Vert^2,
\end{align}

The time-difference $\Delta W(kdt):=W\left(v((k+1)dt)\right)-W\left(v(kdt)\right)$ is
\begin{align}\label{eq:TimeDifference1}
\Delta W(kdt)=&\sum_{\omega=v((k+1)dt)}^{r}\left\Vert \rho\left(\omega,r\right)\right\Vert^2-\sum_{\omega=v(kdt)}^{r}\left\Vert \rho\left(\omega,r\right)\right\Vert^2\nonumber\\
=&-\left\Vert \rho\left(v(kdt),r\right)\right\Vert^2\leq0,
\end{align}
which implies that the Lyapunov function $W\left(v(kdt)\right)$ is decreasing along the AF  denoted by $\rho\left(v(kdt),r\right)$. At this stage, invoking the LaSalle invariance principle \cite{Khalil}, we show that the only entire trajectory of the applied reference that satisfies $\Delta W(kdt)\equiv0$ is $v(kdt)=r$ if $r$ is strictly steady-state admissible, and is $v(kdt)=r^\ast$ otherwise.

% According to \eqref{eq:ERGDiscrete}, it follows from \eqref{eq:TimeDifference1} that
% \begin{align}\label{eq:TimeDifference2}
% \Delta W(k)=&\left\Vert r-v(k)-dt\cdot\kappa\cdot g\left(x(k),v(k),r\right)\right\Vert^2\nonumber\\
% &-\left\Vert r-v(k)\right\Vert^2.
% \end{align}

Consider the following two cases. \\
\noindent$\bullet$ Case I\textemdash $r$ is strictly steady-state admissible: In this case, according to Remark \ref{remark:SteadyStateDelta} and since $\rho_r\left(v(kdt)\right)=\textbf{0},~\forall k$, it is evident from \eqref{eq:AFAttraction} that $\left\Vert \rho\left(v(kdt),r\right)\right\Vert^2=0$ implies that $v(kdt)=r$.

% it follows from \eqref{eq:TimeDifference2} that:
% \begin{align}\label{eq:TimeDifference3}
% \Delta W(k)=&\left\Vert\left(1-\alpha(k)\right)\left(r-v(k)\right)\right\Vert^2-\left\Vert r-v(k)\right\Vert^2\nonumber\\
% =&\left(\left(1-\alpha(k)\right)^2-1\right)\left\Vert r-v(k)\right\Vert^2\leq0,
% \end{align}
% where $\alpha(k):=\frac{dt\cdot\kappa\cdot\Delta\left(x(k),v(k)\right)}{\max\{\left\Vert r-v(k)\right\Vert,\eta\}}\in\mathbb{R}_{\geq0},~\forall k$. According to Remark \ref{remark:SteadyStateDelta}, $\alpha(k)$ cannot remain equal to zero, which implies that the only entire trajectory that satisfies $\Delta W\left(k\right)\equiv0$ is the reference $r$. Thus, $v(k)$ asymptotically converges to $r$.  

% According to Lemma \eqref{Lemma:Condition}, there exists \cite{nicotra2018explicit,hosseinzadeh2019explicit} $\epsilon\in\mathbb{R}_{>0}$ such that $\Delta\left(\bar{x}_{v(k)},v(k)\right)\geq\epsilon$; this property implies that $\Delta\left(x(k),v(k)\right)$ (and consequently $\alpha(k)$) cannot remain equal to zero. 

\begin{figure}[!t]
    \centering
    \includegraphics[width=6.4cm]{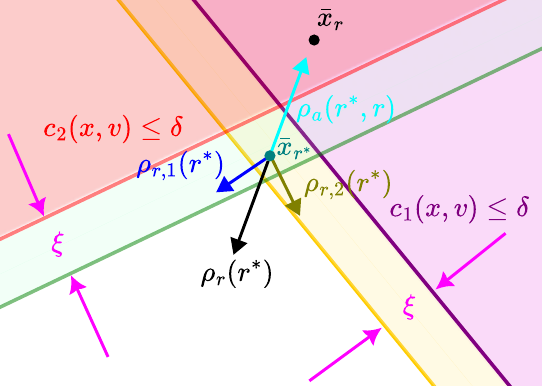}
    \caption{Geometric illustration of $r^\ast$ and the attraction and repulsion terms.}
    \label{fig:Convergence}
\end{figure}

\noindent$\bullet$ Case II\textemdash $r$ is not strictly steady-state admissible: In this case, according to Remark \ref{remark:SteadyStateDelta}, to ensure that $\left\Vert \rho\left(r^\ast,r\right)\right\Vert^2=0$, we only need to show that $r^\ast$ is such that $\rho_a(r^\ast,r)=-\rho_r(r^\ast)$. On the one hand, given a sufficiently small $\eta_1\in\mathbb{R}_{>0}$, it follows from \eqref{eq:AFAttraction} that $\left\Vert\rho_a(v,r)\right\Vert=1,~\forall v\in\mathcal{D}$.  On the other hand, according to \eqref{eq:AFRepulsive}, $\left\Vert\rho_r(v)\right\Vert=1$ can be only satisfied for $v$ satisfying $\delta\leq c_i(\bar{x}_v,v)<\xi$ for some $i$. Thus, since $\rho_r(v)$ is always directed along the constraint gradient, and according to convexity of the set $\mathcal{D}$, there exists a unique $r^\ast$ satisfying $\rho_a(r^\ast,r)=-\rho_r(r^\ast)$. Figure \ref{fig:Convergence} presents a geometric illustration for Case II. 
\end{proof}

Next theorem shows that the discrete-time ERG given in \eqref{eq:ERGDiscrete} guarantees constraint satisfaction at all times.

\begin{theorem}\label{theorem:ConstraintSatisfaction}
Consider the discrete-time ERG given in \eqref{eq:ERGDiscrete}. Let the initial conditions $x(0)$ and $v(0)$ be such that $c_i(x(0),v(0))\geq0,~\forall i$. Then, constraints \eqref{eq:constraints} are satisfied at all times (that is $c_i(x(t),v(t))\geq0,~\forall i$ for $t\geq0$), if the design parameter $\kappa(kdt)$ is selected such that the following inequality is satisfied
\begin{align}\label{eq:ConditionKappa2}
\kappa(kdt)\leq\frac{\max\{\frac{\sqrt{m_1}}{\sqrt{m_1}+\sqrt{m_2}}\bar{\vartheta}-\frac{\sqrt{m_2}}{\sqrt{m_1}+\sqrt{m_2}}\left\Vert x(kdt)-\bar{x}_{v((k-1)dt)}\right\Vert,0\}}{\mu \cdot dt\cdot \max\left\{\left\Vert g(x(kdt),v((k-1)dt),r)\right\Vert,\eta_2\right\}},
\end{align} 
where $\bar{\vartheta}$ and $\eta_2$ are as in Lemma \ref{Lemma:Condition}.

\end{theorem}

\begin{proof}
To show that constraints are satisfied at all times, it is sufficient \cite{hosseinzadeh2019explicit,nicotra2018explicit} to show that $\Delta\left(x(t),v(t)\right)\geq0$ for all $t\geq0$. This proof will be conducted in two parts: i) constraint satisfaction at sampling instants (i.e., the instants at which the applied reference is updated); and ii) constraint satisfaction between sampling instants.

% \begin{figure*}[b]
% \hrule
% \setcounter{equation}{23}
% \begin{align}\label{eq:Condition5}
% & \sqrt{m_2}\left\Vert x(kdt)-\bar{x}_{v((k-1)dt)}\right\Vert+\sqrt{m_2}\left\Vert\bar{x}_{v(kdt)}-\bar{x}_{v((k-1)dt)}\right\Vert\leq\Big|\sqrt{m_1}\left\Vert z^\ast-\bar{x}_{v((k-1)dt)}\right\Vert-\sqrt{m_1}\left\Vert\bar{x}_{v(kdt)}-\bar{x}_{v((k-1)dt)}\right\Vert\Big|.
% \end{align}
% \end{figure*}

Part I\textemdash Once the applied reference is updated, the threshold value $\Gamma\left(vdt\right)$ and the Lyapunov function $V\left(x(kdt),v(kdt)\right)$ may change such that the Lyapunov function is not entirely contained in the constraints anymore; see Figure \ref{fig:SwitchingLyapunov} for a geometric illustration. Thus, at any sampling instant $kdt$, to ensure constraint satisfaction, it is sufficient to determine the applied reference $v(kdt)$ such that the following inequality is satisfied:
\begin{align}\label{eq:Condition1}
V(x(kdt),v(kdt))\leq\min\limits_{i\in\{1,\cdots,n_c\}}\{\hat{\Gamma}_i\left(v(kdt)\right)\},
\end{align}
where\footnote{According to \eqref{eq:GammaOpt1}, it can be concluded that $\min_{i\in\{1,\cdots,n_c\}}\{\hat{\Gamma}_i\left(v(kdt)\right)\}\leq\Gamma\left(v(kdt)\right)$. Thus, the satisfaction of \eqref{eq:Condition1} implies that $V(x(kdt),v(kdt))\leq\Gamma\left(v(kdt)\right)$, or equivalently $\Delta(x(kdt),v(kdt))\geq0$.}
\begin{align}
\hat{\Gamma}_i\left(v(kdt)\right)=\left\{
\begin{array}{cc}
     &  \min\limits_{z\in\mathbb{R}^n} m_1\left\Vert z-\bar{x}_{v(kdt)}\right\Vert^2\\
    \text{s.t.} & c_i(x,v)\leq0,
\end{array}
\right..
\end{align}

Let $i^\ast$ be the index of the smallest $\hat{\Gamma}_i\left(v(kdt)\right)$, i.e., $i^\ast=\arg\min_{i}\{\hat{\Gamma}_i\left(v(kdt)\right)\}$. According to \eqref{eq:Lyapunov}, it is easy to show that \eqref{eq:Condition1} is satisfied if the following inequality is satisfied:
\begin{align}\label{eq:Condition2}
m_2\left\Vert x(kdt)-\bar{x}_{v(kdt)}\right\Vert^2\leq m_1\left\Vert z^\ast-\bar{x}_{v(kdt)}\right\Vert^2,
\end{align}
or equivalently
\begin{align}\label{eq:Condition3}
\sqrt{m_2}\left\Vert x(kdt)-\bar{x}_{v(kdt)}\right\Vert\leq \sqrt{m_1}\left\Vert z^\ast-\bar{x}_{v(kdt)}\right\Vert,
\end{align}
where $z^\ast$ is a point satisfying $c_{i^\ast}(z^\ast,v(kdt))=0$.

Adding and subtracting $\bar{x}_{v((k-1)dt)}$ to the norms in both sides of \eqref{eq:Condition3} yields:
\begin{align}\label{eq:Condition4}
&\sqrt{m_2}\left\Vert x(kdt)-\bar{x}_{v((k-1)dt)}+\bar{x}_{v((k-1)dt)}-\bar{x}_{v(kdt)}\right\Vert\nonumber\\
&\leq \sqrt{m_1}\left\Vert z^\ast-\bar{x}_{v((k-1)dt)}+\bar{x}_{v((k-1)dt)}-\bar{x}_{v(kdt)}\right\Vert. 
\end{align}

\begin{figure}[!t]
    \centering
    \includegraphics[width=8.5cm]{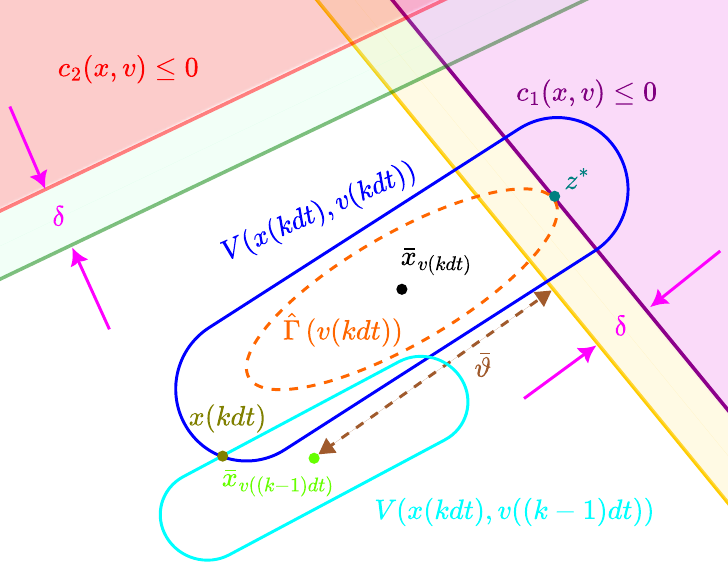}
    \caption{Geometric illustration of changes in the threshold value and the Lyapunov function at sampling instant $kdt$. Note that the Lyapunov function is not necessarily ellipsoidal, while the threshold value is determined based on an ellipsoidal lower-bound of the Lyapunov function as in \eqref{eq:GammaOpt1}.}
    \label{fig:SwitchingLyapunov}
\end{figure}

The triangle inequality implies that \eqref{eq:Condition4} is satisfied if the following inequality\footnote{We use the following inequalities to obtain \eqref{eq:Condition5}: i) $\left\Vert x(kdt)-\bar{x}_{v((k-1)dt)}+\bar{x}_{v((k-1)dt)}-\bar{x}_{v(kdt)}\right\Vert\leq\left\Vert x(kdt)-\bar{x}_{v((k-1)dt)}\right\Vert+\left\Vert\bar{x}_{v(kdt)}-\bar{x}_{v((k-1)dt)}\right\Vert$; and ii) $\Big|\left\Vert z^\ast-\bar{x}_{v((k-1)dt)}\right\Vert-\left\Vert\bar{x}_{v(kdt)}-\bar{x}_{v((k-1)dt)}\right\Vert\Big|\leq\left\Vert z^\ast-\bar{x}_{v((k-1)dt)}+\bar{x}_{v((k-1)dt)}-\bar{x}_{v(kdt)}\right\Vert$.} is satisfied:
\begin{align}\label{eq:Condition5}
& \sqrt{m_2}\left\Vert x(kdt)-\bar{x}_{v((k-1)dt)}\right\Vert+\sqrt{m_2}\left\Vert\bar{x}_{v(kdt)}-\bar{x}_{v((k-1)dt)}\right\Vert\leq\nonumber\\
&~~~~~\Big|\sqrt{m_1}\left\Vert z^\ast-\bar{x}_{v((k-1)dt)}\right\Vert-\sqrt{m_1}\left\Vert\bar{x}_{v(kdt)}-\bar{x}_{v((k-1)dt)}\right\Vert\Big|.
\end{align}

When $\left\Vert\bar{x}_{v(kdt)}-\bar{x}_{v((k-1)dt)}\right\Vert>\left\Vert z^\ast-\bar{x}_{v((k-1)dt)}\right\Vert$, inequality \eqref{eq:Condition5}
leads to contradiction. Thus, we have $\left\Vert\bar{x}_{v(kdt)}-\bar{x}_{v((k-1)dt)}\right\Vert\leq\left\Vert z^\ast-\bar{x}_{v((k-1)dt)}\right\Vert$, which implies that \eqref{eq:Condition5} is satisfied if $\big\Vert\bar{x}_{v(kdt)}-\bar{x}_{v((k-1)dt)}\big\Vert$ satisfies the following inequality:
\begin{align}\label{eq:Condition7}
\left\Vert\bar{x}_{v(kdt)}-\bar{x}_{v((k-1)dt)}\right\Vert\leq&\frac{\sqrt{m_1}}{\sqrt{m_1}+\sqrt{m_2}}\left\Vert z^\ast-\bar{x}_{v((k-1)dt)}\right\Vert\nonumber\\
&-\frac{\sqrt{m_2}}{\sqrt{m_1}+\sqrt{m_2}}\left\Vert x(kdt)-\bar{x}_{v((k-1)dt)}\right\Vert,
\end{align}
where the right-hand side in \eqref{eq:Condition7} is always positive.

% Since $\frac{\sqrt{m_1}}{\sqrt{m_1}+\sqrt{m_2}}<1$, it can be concluded that satisfaction of condition \eqref{eq:Condition7} implies satisfaction of condition \eqref{eq:Condition6}; hence, we conclude that \eqref{eq:Condition5} is satisfied if 
% condition \eqref{eq:Condition7} is satisfied. 

As discussed above, $z^\ast$ is a point on the boundary of the constraint that yields the minimum threshold value when the applied reference is updated; note that determining this point before computing the next applied reference is not possible. However, since $\bar{\vartheta}$ given in \eqref{eq:Distance} is the minimum distance between the equilibrium point $\bar{x}_{v((k-1)dt)}$ and the tightened constraints, it can be concluded that $\bar{\vartheta}\leq\left\Vert z^\ast-\bar{x}_{v((k-1)dt)}\right\Vert$. Thus, one can calim that condition \eqref{eq:Condition7} is satisfied if the following condition is satisfied:
\begin{align}\label{eq:Condition8}
&\left\Vert\bar{x}_{v(kdt)}-\bar{x}_{v((k-1)dt)}\right\Vert\leq\max\Big\{\frac{\sqrt{m_1}}{\sqrt{m_1}+\sqrt{m_2}}\bar{\vartheta}\nonumber\\
&~~~~~~~~~~~~~~~-\frac{\sqrt{m_2}}{\sqrt{m_1}+\sqrt{m_2}}\left\Vert x(kdt)-\bar{x}_{v((k-1)dt)}\right\Vert,0\Big\}.
\end{align}

Finally, according to \eqref{eq:relationship2}, no further effort is needed to show that if the design parameter $\kappa$ satisfies condition \eqref{eq:ConditionKappa2} at sampling instant $kdt$, the condition \eqref{eq:Condition8} is satisfied, which implies that condition \eqref{eq:Condition1} is satisfied.

Part II\textemdash  Consider the time interval $t\in\left[k dt,(k+1)dt\right)$. Note that $v(t)=v(kdt)$ for $t\in\left[kdt,(k+1)dt\right)$, which implies that $\Gamma\left(v(t)\right)=\Gamma\left(v(kdt)\right)$ for $t\in\left[kdt,(k+1)dt\right)$. In other words, we have $\frac{d}{dt}\Gamma(v(t))=0$ for $t\in\left[kdt,(k+1)dt\right)$.

Given $\Delta\left(x(kdt),v(kdt)\right)=\Gamma\left(v(kdt)-V\left(x(kdt),v(kdt)\right)\right)\geq0$, it follows from continuity that $\Delta\left(x(t),v(t)\right)<0$ can only be obtained for $t^\dag\in\left(kdt,(k+1)dt\right)$ if $\Delta\left(x(t^\dag),v(t^\dag)\right)=0$ and $\frac{d}{dt}\Delta\left(x(t),v(t)\right)|_{t=t^\dag}<0$. As mentioned above, $\frac{d}{dt}\Gamma(v(t))|_{t=t^\dag}=0$. Thus, $\frac{d}{dt}\Delta\left(x(t),v(t)\right)|_{t=t^\dag}<0$ implies that $\frac{d}{dt}V(x(t),v(t))|_{t=t^\dag}>0$, which is a contradiction, as system \eqref{eq:system} is pre-stabilized.
\end{proof}

\begin{remark}
Since the right-hand side of \eqref{eq:ConditionKappa2} is lower than that of \eqref{eq:ConditionKappa}, setting $\kappa(kdt)$ such that inequality \eqref{eq:ConditionKappa2} is satisfied ensures that properties of Lemma \ref{Lemma:Condition} and Theorem \ref{theorem:ConstraintSatisfaction} hold.  
\end{remark}

\section{Discrete-Time ERG}\label{sec:MainResults}
This section proposed discrete-time ERG as the main result of this paper, and provides the corresponding pseudocode.

\begin{theorem}
Consider the pre-stabilized system \eqref{eq:system} which is subject to constraints \eqref{eq:constraints}. Suppose that the discrete-time ERG given in \eqref{eq:ERGDiscrete} is utilized to obtain the applied reference at each sampling instant $kdt$, where the design parameter $\kappa(kdt)$ is determined such that the inequality \eqref{eq:ConditionKappa2} is satisfied. Then, the following properties hold: 
\begin{itemize}
    \item constraints \eqref{eq:constraints} are satisfied at all time (i.e., $c_i(x(t),v(t))\geq0,~\forall i$ for $t\geq0$);
    \item if $r$ is \textit{strictly} steady-state admissible (i.e., $c_i(\bar{x}_{r},r)\geq\xi$ for all $i$), then $v(kdt)\rightarrow r$ as $k\rightarrow\infty$;
    \item if $r$ is not strictly steady-state admissible, then $v(kdt)\rightarrow r^\ast$ as $k\rightarrow\infty$, where $r^\ast$ is the best admissible approximation of $r$ satisfying $\delta \leq c_i(\bar{x}_{r^\ast},r^\ast) < \xi$ for some $i$. 
\end{itemize} 
\end{theorem}

\begin{proof}
The proof is a straightforward application of Lemma \ref{Lemma:Condition}, Theorem \ref{theorem:Convergence} and Theorem \ref{theorem:ConstraintSatisfaction}. 
\end{proof}

\subsection{Feature of the Proposed Discrete-Time ERG}
One feature of the proposed discrete-time ERG is that one can easily implement it on computer without having any concern about violating constraints or degrading the performance. Note that the implementation of the proposed discrete-time ERG is very simple, and requires very few changes in the code of the continuous-time ERG. Algorithm \ref{alg:Pseudocode} provides the pseudocode of the proposed discrete-time ERG, where the lines unique to the proposed discrete-time ERG are shown in blue.

\begin{algorithm}[!t]
\caption{Pseudocode of the Proposed Discrete-Time ERG}\label{alg:Pseudocode}
\textbf{Input}: $x(kdt)$, $v\left((k-1)dt\right)$, $dt$, $m_1$, $m_2$, $\eta_1$, $\eta_2$, $\delta$, $\xi$, and $r$
\begin{algorithmic}[1]
\STATE {$\Gamma \gets \left\{
\begin{array}{cc}
     &  \min\limits_{z\in\mathbb{R}^n} V(z,v)\\
    \text{s.t.} & c_i(z,v)\leq0,~\forall i
\end{array}
\right.$}
\STATE {$V \gets V\left(x(kdt),v\left((k-1)dt\right)\right)$}
\STATE {$\rho_a \gets \frac{r-v}{\max\{\left\Vert r-v\right\Vert,\eta_1\}}$}
\STATE {$\rho_r \gets \sum\limits_{i=1}^{n_c}\max\left\{\frac{\xi-c_i(\bar{x}_v,v)}{\xi-\delta},0\right\}\frac{\nabla_vc_i(\bar{x}_v,v)}{\left\Vert\nabla_vc_i(\bar{x}_v,v)\right\Vert}$}
\STATE {$g \gets \left(\Gamma-V\right)\cdot\left(\rho_a+\rho_r\right)$}
{\color{blue}
\FOR{$i \gets 1 $ to $n_c$}
\STATE{$\vartheta_i\gets\left\{
\begin{array}{ll}
     & \min\,\left\Vert \bar{x}_{v((k-1)dt)}-\theta\right\Vert \\
    \text{s.t.} & c_i(\theta,v((k-1)dt))=\delta
\end{array}
\right.$}
\ENDFOR
\STATE{{\color{blue}$\vartheta \gets \min\limits_{i\in\{1,\cdots,n_c\}}\{\vartheta_i\}$}}
\STATE{$\kappa\gets \frac{\max\{\frac{\sqrt{m_1}}{\sqrt{m_1}+\sqrt{m_2}}\bar{\vartheta}-\frac{\sqrt{m_2}}{\sqrt{m_1}+\sqrt{m_2}}\left\Vert x(kdt)-\bar{x}_{v((k-1)dt)}\right\Vert,0\}}{\mu \cdot dt\cdot \max\left\{\left\Vert g\right\Vert,\eta_2\right\}}$}
}
\STATE{$v\gets v((k-1)dt)+dt\cdot\kappa\cdot g$}
\RETURN{$v$}
\end{algorithmic}
\end{algorithm}

Another feature of the proposed discrete-time ERG is that it performs well in the tradeoff space that involves performance and constraint satisfaction. In particular, the proposed discrete-time ERG uses a dynamic $\kappa$, which has some key advantages over a constant $\kappa$, such as: i) there is no need to perform offline computations to determine a constant $\kappa$ that ensures constraint satisfaction; ii) the proposed discrete-time ERG uses a large $\kappa$ whenever it does not lead to constraint violation, which means that the proposed discrete-time ERG does not unnecessarily degrade the performance when the constraints are far from being violated.

% \begin{remark}
% According to Lemma \ref{Lemma:Condition} and Theorem \ref{theorem:ConstraintSatisfaction}, it can be concluded that the set $\mathcal{D}$ remains invariant and constraints \eqref{eq:constraints} are satisfied at all times, if the design parameter $\kappa$ satisfies condition \eqref{eq:ConditionKappa2} at any sampling instant. 
% \end{remark}

% \begin{remark}
%  In condition \eqref{eq:Condition7}, the expression is always positive. This is a consequence of the fact that the distance between the equilibrium point and the constraint, denoted as $\left\Vert z^\ast-\bar{x}_{v(k-1)}\right\Vert$, must be less than the distance between the equilibrium point $\bar{x}_{v(k-1)}$ and the current state $x(k)$. Violation of this condition would imply a contradiction, as it would indicate that the system has deviated from the desired equilibrium point by a greater distance than the minimum allowed by the constraints.
% \end{remark}

%%%%%%%%%%%%%%%%%%%%%%%%%%%%%%%%%%%%%%%%%%%

\section{Simulation Results}\label{sec:NA}
This section aims at evaluating the effectiveness of the proposed discrete-time ERG via extensive simulation studies on two examples: i) controlling a double integrator; and ii) controlling the longitudinal dynamics of an aircraft.

\subsection{Controlling a Double Integrator}
Consider the system $\ddot{x}(t)=u(t)$, which is subject to the constraint $x(t)\leq1,~\forall t\geq0$. This system can be prestabilized using a full state feedback control law as $u(t) = -K[x(t)~\dot{x}(t)]^\top+ Gv(t)$, where $K=\left[\begin{matrix}10 & 0.5\end{matrix}\right]$ is the feedback gain vector and and $G=10$ is the feedforward gain. Stability of the pre-stabilized system can be shown by using the Lyapunov function $V(x,\dot{x},v)=\begin{bmatrix}x-v  & \dot{x}\end{bmatrix}\begin{bmatrix}2.25 & -1\\-1 & 22 \end{bmatrix}\begin{bmatrix}x-v \\ \dot{x}\end{bmatrix}$, which implies that 
$m_1=2.2$ and $m_2=22$. It can be easily shown that $\left\Vert\nabla_v\bar{x}_v\right\Vert=\left\Vert\nabla_v[v~0]^\top\right\Vert=1$, which implies that $\mu=1$. Also, the set of steady-state admissible equilibria is $\mathcal{D}=\{v|v\leq1-\delta\}$.

We assume that the desired reference is $r=1.1$. We update the applied reference every $dt=0.1$ seconds, and use the following parameters to implement the discrete-time ERG: $\eta_1=0.01$, $\eta_2=0.01$, $\xi=0.045$, and $\delta=0.04$.  We use the method described in \cite{garone2018explicit,HosseinzadehECC} to compute the threshold value $\Gamma(v)$.

To assess the performance of the proposed discrete-time ERG in guaranteeing constraint satisfaction at all times, we consider five cases based on the choice of the parameter $\kappa$, including $\kappa=1$ as suggested in \cite{HosseinzadehLetter}. To provide a quantitative comparison, we consider 20,000 experiments with initial condition $x(0)=[\beta~0]^\top$, where in each experiment, $\beta$ is uniformly selected from the interval [-50,0.95]; to ensure a fair comparison, the same initial condition is applied to all five cases.

\begin{table}[t]
\centering
\caption{Double Integrator\textemdash Comparison Study}
\label{table:Comparison1}
\begin{tabular}{c|c|c}
\hline
Method   & Violation of  & Violation of  \\ 
& the Constraint  & the Invariance of   \\
& $x\leq1$  & the Set $\mathcal{D}$    \\
\hline
Proposed Discrete-   & 0 \% & 0 \%  \\ 
Time ERG    &   &  \\ 
\hline
$\kappa=0.1$    & 43.36\%    &  9.41\%  \\ 
\hline
$\kappa=0.4$    & 74.96\%   &  66.72\%  \\ 
\hline
$\kappa=0.7$    &  79.50\%   & 74.79\%  \\ 
\hline
$\kappa=1.0$  & 81.34\%   &  78.15\%  \\ 
(suggested in \cite{HosseinzadehLetter})   &   &  \\ 
\hline
\end{tabular}
\end{table}

The obtained results are summarized in TABLE \ref{table:Comparison1}, which shows the effectiveness of the proposed method in guaranteeing constraint satisfaction and ensuring invariance of set $\mathcal{D}$, when the updates of the applied reference are performed in discrete time as in \eqref{eq:ERGDiscrete}. As seen in TABLE \ref{table:Comparison1} when the ERG is implemented with a constant $\kappa$, discretization errors can lead to violation of constraints/invariance of set $\mathcal{D}$.

\begin{figure}[t]
    \centering
    \includegraphics[width=8.5cm]{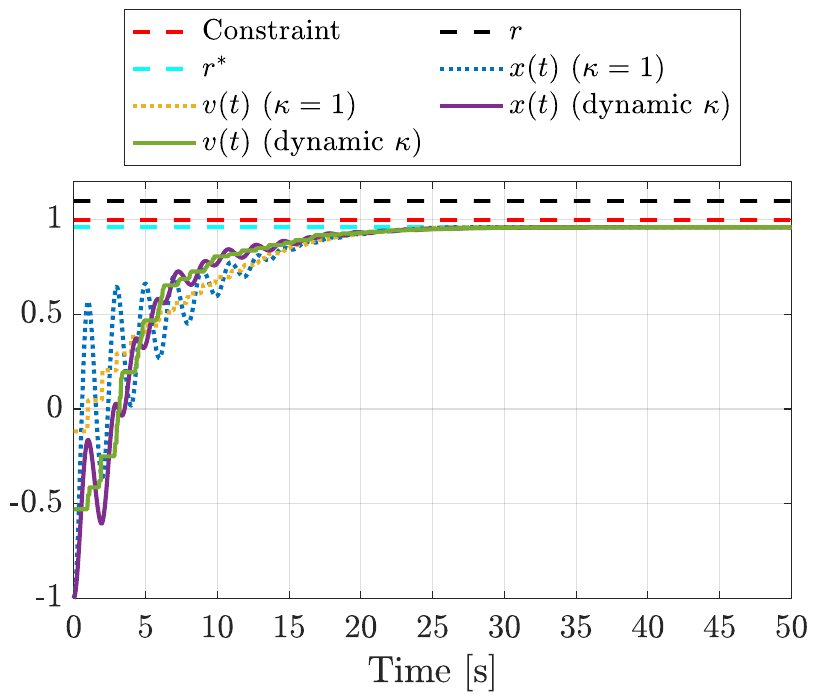}\\
    \includegraphics[width=\columnwidth]{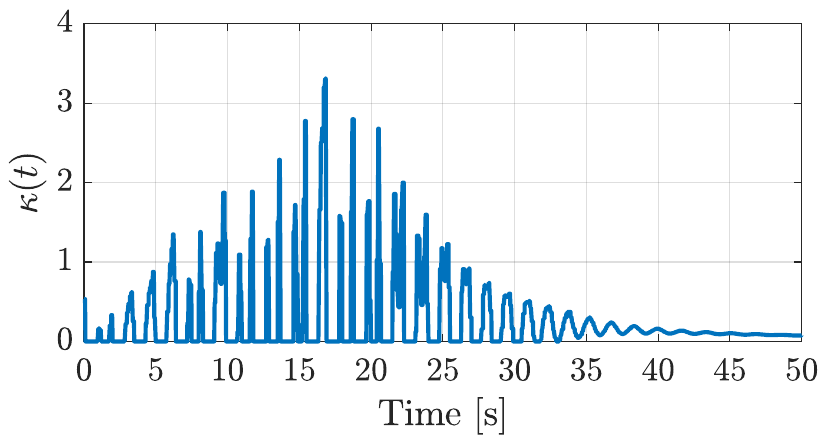}
    \caption{Top: simulation results with the proposed discrete-time ERG and with the ERG scheme with constant $\kappa=1$ (suggested in \cite{HosseinzadehLetter}); bottom: time profile of the determined $\kappa$ with the proposed discrete-time ERG.}
    \label{fig:criticalpoint}
\end{figure}

To compare the case of dynamic $\kappa$ with that of constant $\kappa$ with respect to the convergence performance defined as $\left\Vert x(t)-r\right\Vert$, we consider the initial condition $x(0)=[-1~0]^\top$ and $v(0)=-1$, which is a feasible point when $\kappa=1$. Figure \ref{fig:criticalpoint} shows the simulation results.  As seen in this figure, since the proposed discrete-time ERG increases $\kappa$ whenever needed, it converges faster than the scheme with constant $\kappa=1$; more precisely, the proposed discrete-time ERG has 9.16\% gain over the ERG with $\kappa=1$ with respect to the convergence performance.

Note that although $\kappa$ converges to 0.1 with the proposed discrete-time ERG, as reported in Table \ref{table:Comparison1}, keeping $\kappa=0.1$ can lead to constraint violation in 43.36\% of the experiments. This highlights the fact that there is no need to perform offline computations to determine a safe $\kappa$, and the proposed discrete-time ERG determines a safe $\kappa$ at any sampling instant automatically.

\subsection{Controlling the Longitudinal Dynamics of an Aircraft}
This subsection considers the problem of controlling the longitudinal dynamics of an aircraft; see Figure \ref{fig:aircraft}. The longitudinal dynamics of an aircraft can be expressed as \cite{Nicotra2015_2}:
\begin{equation}\label{eq:longitudinal dynamics}
J\ddot{\alpha} = -d_1L(\alpha) \cos(\alpha) - b \dot{\alpha} + d_2u \cos(\alpha),\end{equation}
where $\alpha$ is the angle of attack, $d_1=4$ [m] and $d_2=42$ [m] are the distances between the center of mass and the two airfoils, $b=2\times10^6$ [Nms/rad] is an estimated viscous friction coefficient, $J=4.5\times10^6$ [kgm$^2$] is the longitudinal inertia of the aircraft, $L(\alpha)=2.5\times10^5+1.5\times10^5\alpha-230\alpha^3$ is the lift generated by the main wing, and $u$ is the control force generated by the elevator airfoil. To prevent the main wing from stalling, the angle of attack is subject to the constraint $\alpha<\alpha_s$, where $\alpha_s=14.7$ [deg].

\begin{figure}[!t]
    \centering
    \includegraphics[width=8.5cm]{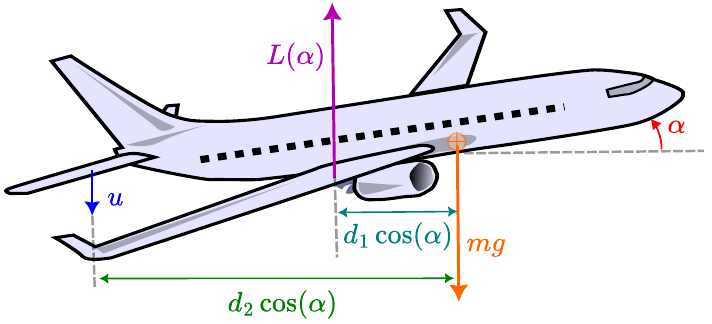}
    \caption{Longitudinal representation of a civil aircraft}
    \label{fig:aircraft}
\end{figure}

This system can be pre-stabilized using the PD control law $u = -k_P\left(\alpha-v\right)-k_D\dot{\alpha}+\frac{d_1}{d_2}L\left(v\right)$, where $k_P=4.7\times10^5$ and $k_D=1.79\times10^5$ are proportional and derivative gains, respectively. Stability of the pre-stabilized system can be shown by using the following Lyapunov function:
\begin{align}
V(\alpha,\dot{\alpha},v)=\frac{1}{2}J\dot{\alpha}^2+\int_\alpha^v\Big(&d_1\left(L(\beta)-L(v)\right)\nonumber\\
&+d_2k_P\left(\beta-v\right)\Big)\cos(\beta)d\beta,
\end{align}
which can be \cite{nicotra2018explicit} lower- and upper-bounded as follows:
\begin{align}
\begin{bmatrix} \alpha-v\\\dot{\alpha}\end{bmatrix}^\top P_1\begin{bmatrix} \alpha-v\\\dot{\alpha}\end{bmatrix}\leq V(\cdot)\leq\begin{bmatrix} \alpha-v\\\dot{\alpha}\end{bmatrix}^\top P_2\begin{bmatrix} \alpha-v\\\dot{\alpha}\end{bmatrix},
\end{align}
where
\begin{align}
P_1=&\begin{bmatrix}
   \frac{1}{2}\left(d_1\frac{L(\alpha_s)-L(v)}{\alpha_s-v}+d_2k_P\right)\cos(\alpha_s) & 0\\ 
    0 & \frac{1}{2}J
\end{bmatrix},\\
P_2=&\begin{bmatrix}
   \frac{1}{2}d_2k_P+\frac{1}{2}d_1\nabla_v L(v) & 0\\ 
    0 & \frac{1}{2}J
\end{bmatrix}.
\end{align}

We assume that $\alpha(0)=\dot{\alpha}(0)=0$, and the desired angle of attack is $r=14$ [deg]. We update the applied reference every $dt=0.1$ seconds, and use the following parameters to implement the discrete-time ERG: $\eta_1=0.01$, $\eta_2=0.01$, $\xi=0.3$, and $\delta=0.1$.

Simulation results for three cases based on the choice of the parameter $\kappa$ are shown in Figure \ref{fig:aircraftResults}: i) $\kappa=10^{-3}$ which is suggested in \cite{nicotra2018explicit}; ii) dynamic $\kappa$ obtained as in \eqref{eq:ConditionKappa2} at every sampling instant; and iii) $\kappa=10^{-9}$ which is the value of $\kappa$ at the beginning of the process with the proposed discrete-time ERG.

As seen in Figure \ref{fig:aircraftResults}, $\kappa=10^{-3}$ which is suggested in \cite{nicotra2018explicit} leads to constraint violation. Implementing ERG with $\kappa=10^{-9}$ does not cause constraint violation; however, the obtained convergence performance is poor. The proposed discrete-time ERG starts with $\kappa=10^{-9}$ and increases the value of $\kappa$ whenever it is safe; as a result, the proposed discrete-time ERG not only guarantees constraint satisfaction at all times, but also yields a better convergence performance compared with cases with constant $\kappa$.

\begin{figure}[!t]
    \centering
    \includegraphics[width=8.5cm]{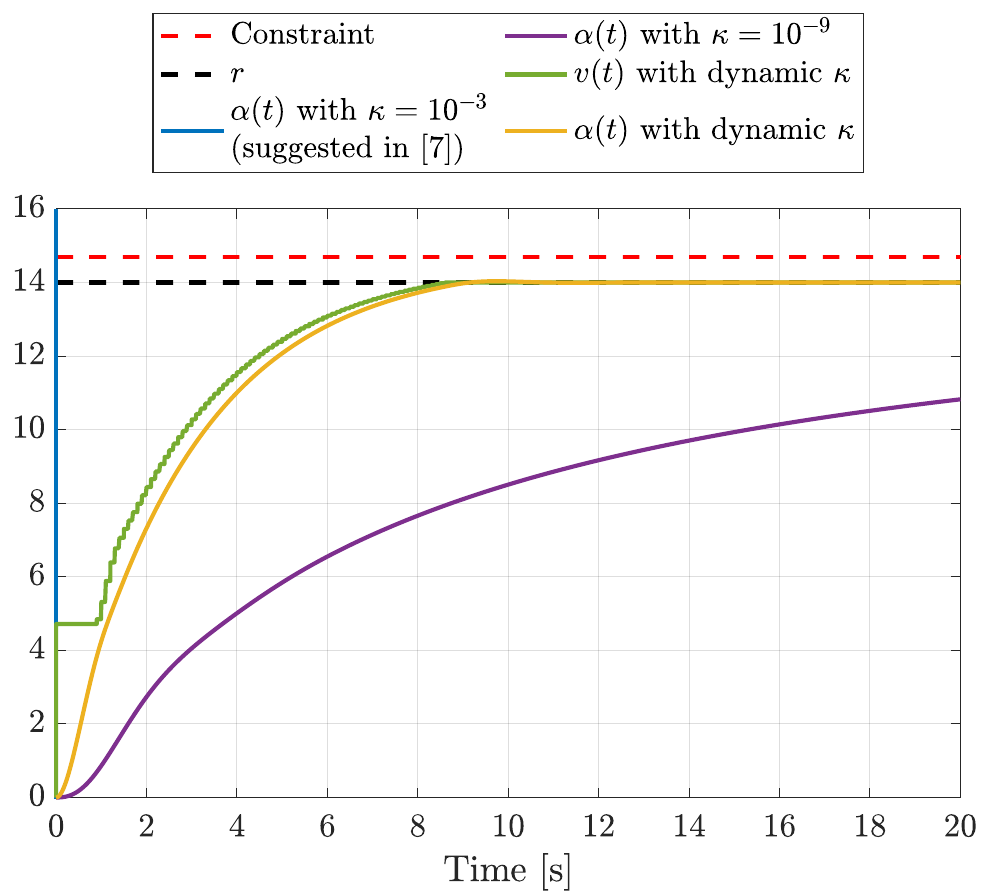}
    \includegraphics[width=8.5cm]{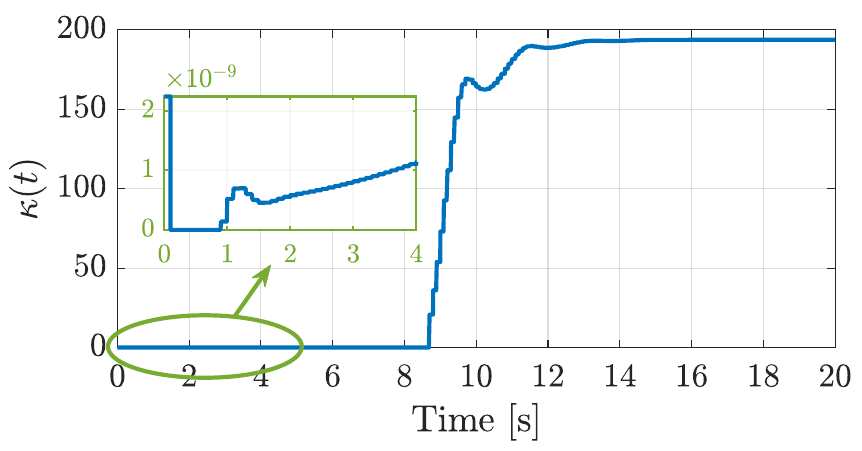}
    \caption{Top: simulation results with the proposed discrete-time ERG and with the ERG scheme with $\kappa=10^{-3}$ (suggested in \cite{nicotra2018explicit}), and $\kappa=10^{-9}$;  bottom: time profile of the determined $\kappa$ with the proposed discrete-time ERG.}
    \label{fig:aircraftResults}
\end{figure}

\section{Experimental Results}\label{sec:ExperimentalResults}
This section aims at experimentally validating the proposed discrete-time ERG by using it to navigate a Parrot Bebop 2 drone. Our experimental setup is shown in Figure \ref{fig:network}. We use \texttt{OptiTrack} system with ten \texttt{Prime$^\text{x}$ 13} cameras operating at a frequency of 120 Hz; these cameras provide a 3D accuracy of $\pm0.02$ millimeters which is acceptable for identification purposes. The computing unit is a $13^{\text{th}}$ Gen $\text{Intel}^{\text{\textregistered}}$ $\text{Core}^{\text{\texttrademark}}$ i9-13900K processor with 64GB RAM, on which the software \texttt{Motive} is installed to analyze and interpret the camera data. We use the ``Parrot Drone Support from MATLAB" package \citep{MATLAB}, and send the control commands to the Parrot Bebop 2 via WiFi and by using the command \texttt{move($\cdot$)}. It should be mentioned that the communication between \texttt{Motive} and MATLAB is established through User Datagram Protocol (UDP) communication using the \texttt{NatNet} service.

\begin{figure}[!t]
    \centering
    \includegraphics[width=8.5cm]{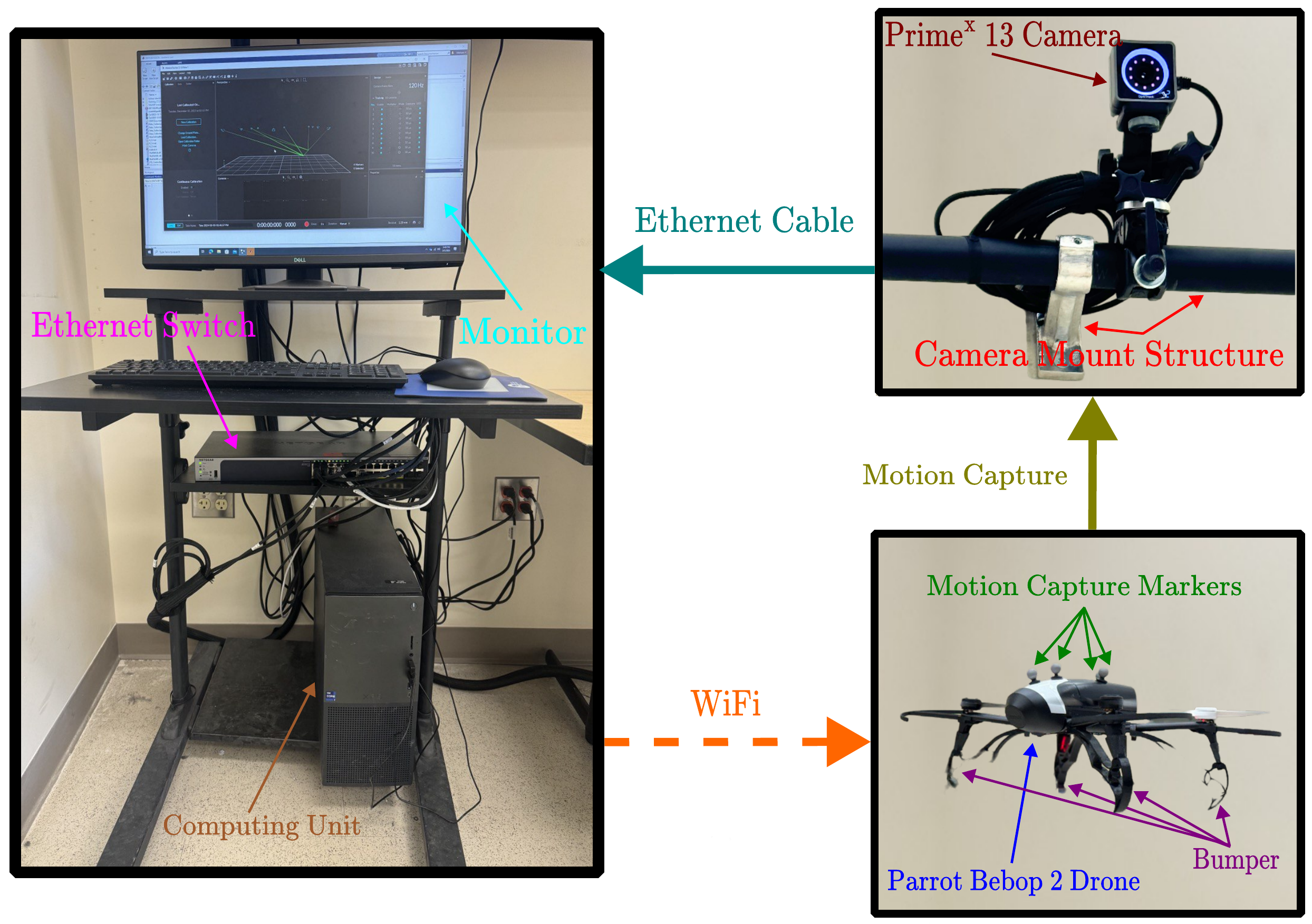}
    \caption{Overview of the experimental setup utilized to experimentally validate the proposed discrete-time ERG.}
    \label{fig:network}
\end{figure}

The dynamical model of the Parrot Bebop 2 drone can be expressed \cite{AmiriMECC} as $\dot{x}=Ax+Bu$, where $x=[p_x~\dot{p}_x~p_y~\dot{p}_y~p_z~\dot{p}_z]^\top$ with $p_x$, $p_y$, and $p_z$ being  X, Y, and Z positions of the drone in the global Cartesian coordinate system, $u=[u_x~u_y~u_z]^\top$ with $u_x$ being the pitch angle, $u_y$ being the roll angle, and $u_z$ being the vertical velocity, and 
\begin{align}
A=&\begin{bmatrix}
0 & 1 & 0 & 0 & 0 & 0 \\
0 & -0.05  & 0 & 0 & 0 & 0 \\
0 & 0 & 0 & 1 & 0 & 0 \\
0 & 0 & 0 &-0.02 & 0 & 0 \\
0 & 0 & 0 & 0 & 0 & 1 \\
0 & 0 & 0 & 0 & 0 & -1.79
\end{bmatrix},\\
B=&\begin{bmatrix}
0 & -5.48 & 0 & 0 & 0 &0 \\
0 & 0 & 0 & -7.06 & 0 & 0 \\
0 & 0 & 0& 0 & 0 & -1.74 
\end{bmatrix}^\top.
\end{align}

We use the following PD control laws to pre-stabilize the dynamics of the Parrot Bebop 2 drone\footnote{We use the PD control law $u_\psi=-\psi-0.8\dot{\psi}$ to regulate the yaw angle of the Parrot Bebop 2 drone to zero.}:
\begin{align}
u_x=&-0.08\left(p_x-v_x\right)-0.06\dot{p}_x,\\
u_y=&-0.08\left(p_y-v_y\right)-0.06\dot{p}_y,\\
u_z=&-1.7\left(p_z-v_z\right)-0.05\dot{p}_z,
\end{align}
where $v_x$, $v_y$, and $v_z$ are the applied references in X, Y, and Z directions, respectively. The stability of the pre-stabilized system can be shown by using the following Lyapunov function:
\begin{align}
V(x,v)=\begin{bmatrix} p_x-v_x\\ \dot{p}_x \\p_y-v_y \\\dot{p}_y \\p_z-v_z \\ \dot{p}_z
\end{bmatrix}^\top\begin{bmatrix}P_1 & \mathbf{0} & \mathbf{0}\\
\mathbf{0} & P_2 & \mathbf{0}\\
\mathbf{0} & \mathbf{0} & P_3
\end{bmatrix}\begin{bmatrix} p_x-v_x\\ \dot{p}_x \\p_y-v_y \\\dot{p}_y \\p_z-v_z \\ \dot{p}_z
\end{bmatrix},
\end{align}
where
\begin{align}
P_1=\begin{bmatrix}9.48 & -1 \\
-1 & 3.77 
\end{bmatrix},~P_2=\begin{bmatrix}7.05 & -1 \\
-1 & 3.54
\end{bmatrix},~P_3=\begin{bmatrix}1.35 & -1 \\
-1 & 2.11 
\end{bmatrix},
\end{align}
and thus, $m_1=0.6592$ and $m_2=9.6460$.

We assume that the position of the Parrot Bebop 2 drone along Y direction should satisfy the constraint $p_y(t)<1,~\forall t\geq0$. We assume that the desired positions along X, Y, and Z directions are 0, 1.2, and 1.5, respectively. To implement the proposed discrete-time ERG, we assume that $dt=0.083$ seconds, and we set $\eta_1=0.01$, $\eta_2=0.01$, $\xi=0.045$, and $\delta=0.04$.

We assume that the desired position along Y direction is $r=1.2$, which is not steady-state admissible. Figure \ref{fig:droneResult} shows the experimental results for three different scenarios based on the parameter $\kappa$: i) $\kappa=0.4$; ii) dynamic $\kappa$ obtained as in \eqref{eq:ConditionKappa2} at every sampling instant; and iii) $\kappa=0.08$ which is the value of $\kappa$ at the beginning of the process with the proposed discrete-time ERG. Note that due to space limitations, we only present time profile of the Parrot Bebop 2 drone's position along Y direction.

As seen in Figure \ref{fig:droneResult}, $\kappa=0.4$ leads to constraint violation. While $\kappa=0.08$ does not result in a constraint violation when using ERG, the convergence performance that is obtained is poor. In contrast to cases where $\kappa$ is constant, the suggested discrete-time ERG yields a better convergence performance and ensures that the constraint $p_y(t)<1$ is satisfied at all times. As seen in Figure \ref{fig:droneResult}-bottom, the proposed scheme begins with $\kappa=0.08$ and increases the value of $\kappa$ whenever it is safe to do so.

\section{Conclusion}\label{sec:Conclusion}
The ERG is a simple and systematic approach that provides constraint handling capability to pre-stabilized systems. The current literature on ERG provides theoretical analyses that are developed in continuous time. However, ERG must be implemented in discrete time to enforce constraints in real-world applications. This paper studied theoretical guarantees for the discrete-time implementation of ERG. It was analytically shown, when the updates of the applied reference are computed in discrete time, it is possible to maintain the feasibility and convergence properties of the ERG framework by appropriately adjusting the design parameters. The effectiveness of the proposed approach has been assessed via extensive simulation and experimental studies.

\begin{figure}[!t]
    \centering
    \includegraphics[width=8.5cm]{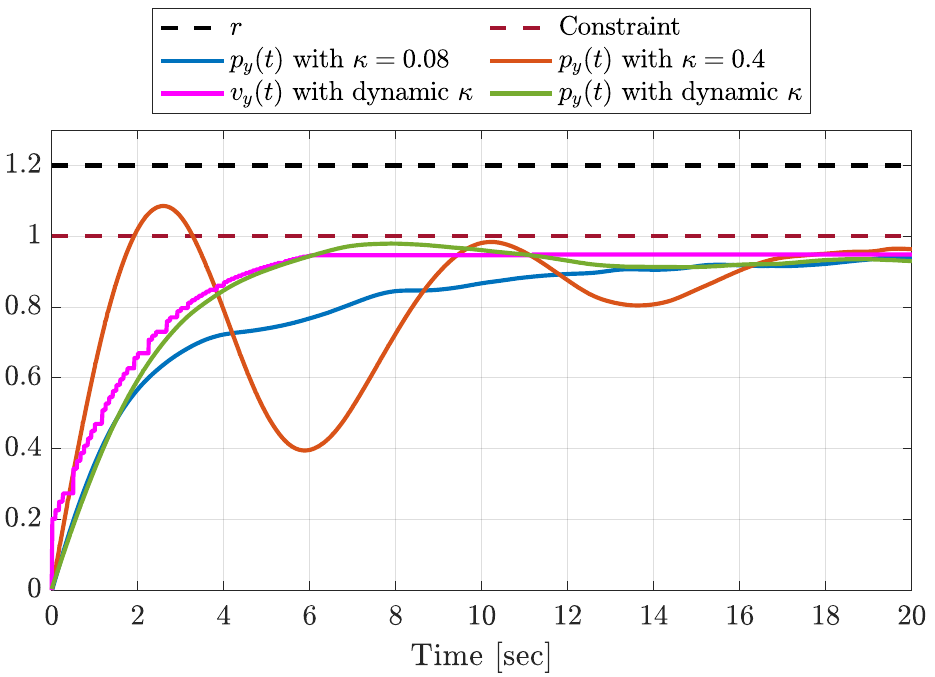}
    \includegraphics[width=8.5cm]{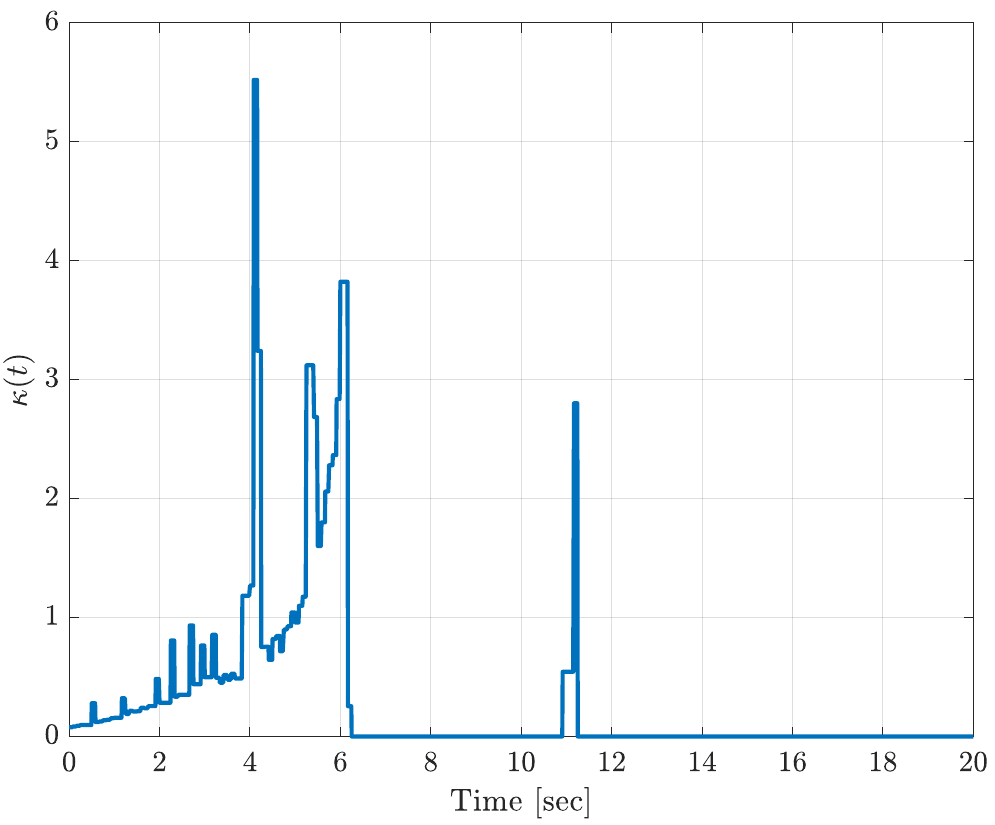}
    \caption{Top: experimental results with the proposed discrete-time ERG and with the ERG scheme with $\kappa=0.08$, and $\kappa=0.4$;  bottom: time profile of the determined $\kappa$ with the proposed discrete-time ERG.}
    \label{fig:droneResult}
\end{figure}

%\balance
\bibliographystyle{elsarticle-num}
\bibliography{Reference}

\end{document}